\documentclass[a4paper,UKenglish,cleveref, autoref, thm-restate]{lipics-v2021}

\title{Efficient computation of topological integral transforms} 

\author{Vadim {Lebovici}}{Mathematical Institute, University of Oxford, United Kingdom \and \url{https://vadimlebovici.github.io/}}{vadim.lebovici@maths.ox.ac.uk}{https://orcid.org/0000-0002-2709-3919}{}

\author{Steve {Oudot}}{Inria Saclay, France\and \url{https://geometrica.saclay.inria.fr/team/Steve.Oudot/}}{steve.oudot@inria.fr}{}{}

\author{Hugo {Passe}}{Ecole Normale Supérieure de Lyon, France\and \url{}}{hugo.passe@ens-lyon.fr}{}{}

\authorrunning{V. Lebovici, S. Oudot, H. Passe} 

\Copyright{Vadim Lebovici, Steve Oudot, Hugo Passe} 



\category{} 

\relatedversion{} 

\supplement{https://github.com/HugoPasse/Eucalc}

\nolinenumbers 

\EventEditors{Leo Liberti}
\EventNoEds{1}
\EventLongTitle{22nd International Symposium on Experimental Algorithms (SEA 2024)}
\EventShortTitle{SEA 2024}
\EventAcronym{SEA}
\EventYear{2024}
\EventDate{July 23-26, 2024}
\EventLocation{Vienna, Austria}
\EventLogo{}
\SeriesVolume{301}
\ArticleNo{23}

\usepackage[dvipsnames]{xcolor} 

\usepackage{amsmath,amssymb}
\usepackage{url}
\usepackage{multirow}

\usepackage{graphicx}

\usepackage{caption}
\usepackage{subcaption}

\usepackage{algorithm}
\usepackage{algpseudocode}
\algnewcommand{\LeftComment}[1]{\Statex \(\triangleright\) #1}
\algnewcommand{\Method}[1]{\Statex \textbf{method} #1}

\usepackage{hyperref}

\usepackage{xparse} 

\numberwithin{equation}{section}

\def\Z{{\mathbb Z}}    
\def\R{{\mathbb R}}    
\def\C{{\mathbb C}}    

\def\V{\mathbb{V}}

\def\kk{\R}

\def\st{\, : \,} 
\newcommand{\card}[1]{\##1}
\def\bigO{\mathcal{O}}

\def\sgn{\mathrm{sgn}}

\NewDocumentCommand{\Lloc}{O{\R} O{1}}{%
\mathrm{L}^{#2}_{\textnormal{loc}}(#1)%
}
\NewDocumentCommand{\Lp}{O{\R} O{1}}{%
\mathrm{L}^{#2}(#1)%
}

\newcommand{\eps}{\varepsilon} 
\renewcommand{\phi}{\varphi} 

\DeclareMathOperator*{\limit}{lim}

\newcommand{\relint}[1]{\mathrm{relint}(#1)}

\renewcommand{\lim}[1][ ]{\varprojlim_{#1}}

\def\field{\kk}

\newcommand{\dual}[1]{{#1}^{*}}
\newcommand{\dualdot}[2]{\langle #1; #2\rangle}
\NewDocumentCommand{\projective}{d[]}{
\IfNoValueTF{#1}{\def\argument{}}{\def\argument{\hspace{-0.1em}\left(#1\right)}}%
\mathbb{P}\argument%
}
\NewDocumentCommand{\dprojective}{d[]}{
\IfNoValueTF{#1}{\def\argument{}}{\def\argument{\hspace{-0.2em}\left(#1\right)}}%
\dual{\mathbb{P}}\argument%
}

\renewcommand{\vert}[1]{\mathrm{Vert(#1)}}
\def\cplx{\mathcal{P}}
\def\ccplx{\mathcal{C}}
\def\collec{\mathcal{X}}

\def\cplxleq{\cplx_{\{\xi\leq t\}}}
\def\cplxeq{\cplx_{\{\xi=t\}}}
\def\phileq{\phi_{\{\xi\leq t\}}}
\def\phieq{\phi_{\{\xi=t\}}}

\newcommand{\sgnvec}[1]{\eps^{#1}}

\def\cube{P}
\def\cubee{Q}

\NewDocumentCommand{\dist}{O{\V} O{x}}{%
\mathcal{D}_{#2}'(#1)%
}
\NewDocumentCommand{\tempdist}{O{\V} O{x}}{%
\mathcal{S}_{#2}'(#1)%
}
\NewDocumentCommand{\schwartz}{O{\V} O{x}}{%
\mathcal{S}_{#2}(#1)%
}
\NewDocumentCommand{\Fourier}{m O{x}}{%
\mathcal{F}_{#2}(#1)%
}

\def\Laplace{\mathcal{L}}
\def\Fourier{\mathcal{F}}
\def\Bessel{\mathcal{B}}

\def\PL{\textnormal{PL}}        					

\NewDocumentCommand{\Mod}{O{} O{\V} O{\field}}{%
\textnormal{Mod}_{#1}({#3}_{#2})%
}
\NewDocumentCommand{\Db}{O{} O{\V} O{\mathbf{k}}}{%
\textnormal{D}^{\textnormal{b}}_{#1}({#3}_{#2})%
}
\def\dt{\mathrm{d}t}

\def\conv{\star}

\NewDocumentCommand{\CF}{O{\PL} O{\R^n}}{%
\textnormal{CF}_{#1}(#2)%
}

\def\1{{\mathbf{1}}} 
\NewDocumentCommand{\Euler}{d[]}{
\IfNoValueTF{#1}{\def\argument{}}{\def\argument{\left(#1\right)}}%
\chi\argument%
}
\NewDocumentCommand{\Eulerc}{d[]}{
\IfNoValueTF{#1}{\def\argument{}}{\def\argument{\left(#1\right)}}%
\chi_c\argument%
}
\NewDocumentCommand{\EPind}{d[]}{
\IfNoValueTF{#1}{\def\argument{}}{\def\argument{\left(#1\right)}}%
\chi_{\mathrm{loc}}\argument%
}

\NewDocumentCommand{\epigraph}{O{f} O{}}{%
\Gamma_{#1}^{#2}%
}
\NewDocumentCommand{\sublevelCF}{m O{\gamma}}{%
\phi_{#1}^{#2}%
}

\newcommand{\lwcritvar}[2][\phi]{\Delta_{#1}^-(#2)}
\newcommand{\ordcritvar}[2][\phi]{\Delta_{#1}^{\mathrm{ord}}(#2)}
\newcommand{\lwcritpts}[2][\phi]{\mathrm{Crit}_{#1}^-(#2)}
\newcommand{\ordcritpts}[2][\phi]{\mathrm{Crit}_{#1}^{\mathrm{ord}}(#2)}
\newcommand{\critpts}[2][\phi]{\mathrm{Crit}_{#1}(#2)}
\def\codelwcritpts{C^-}
\def\codeordcritpts{C^{\mathrm{ord}}}
\def\codelwcritval{\Delta^-}
\def\codeordcritval{\Delta^{\mathrm{ord}}}

\newcommand{\str}[1]{\mathrm{St}(#1)}
\newcommand{\lwst}[1]{\mathrm{LwSt}(#1)}
\newcommand{\upst}[1]{\mathrm{UpSt}(#1)}
\newcommand{\lwstpt}[1]{\mathrm{LwSt}^*(#1)}
\newcommand{\upstpt}[1]{\mathrm{UpSt}^*(#1)}

\newcommand{\Radon}[1][\phi]{\mathcal{R}_{#1}}

\newcommand{\ECT}[1][\phi]{\mathrm{ECT}_{#1}}
\def\kernel{\kappa}
\def\bkernel{\bar{\kappa}}
\newcommand{\HT}[1][\phi]{\mathrm{HT}_{#1}}

\NewDocumentCommand{\transform}{s O{\phi} d<> O{\kappa}}{%
	\IfBooleanTF #1%
	{\mathrm{T}_{#4}}%
	{\IfNoValueTF{#3}{\def\argument{}}{\def\argument{\left(#3\right)}}%
	\mathrm{T}_{#4}\left[#2\right]{\argument}
	}%
}

\NewDocumentCommand{\subleveltransform}{O{\kappa} d<> O{f_{|Z}}}{%
\IfNoValueTF{#2}{\def\argument{}}{\def\argument{\left(#2\right)}}%
\mathrm{Sub}_{#1}\left[#3\right]{\argument}%
}

\def\EL{\mathcal{E}\hspace{-0.1em}\Laplace}
\NewDocumentCommand{\ELaplace}{O{\phi} d<>}{%
\IfNoValueTF{#2}{\def\argument{}}{\def\argument{\left(#2\right)}}%
\EL\left[#1\right]{\argument}%
}

\def\EF{\mathcal{E}\hspace{-0.1em}\Fourier}
\NewDocumentCommand{\EFourier}{O{\phi} d<>}{%
\IfNoValueTF{#2}{\def\argument{}}{\def\argument{\left(#2\right)}}%
\EF\left[#1\right]{\argument}%
}
\NewDocumentCommand{\GREFourier}{O{\phi} d<>}{%
\IfNoValueTF{#2}{\def\argument{}}{\def\argument{\left(#2\right)}}%
\EF^\mathrm{GR}\left[#1\right]{\argument}%
}
\def\EB{\mathcal{E}\Bessel}
\NewDocumentCommand{\EBessel}{O{\phi} d<>}{%
\IfNoValueTF{#2}{\def\argument{}}{\def\argument{\left(#2\right)}}%
\EB\left[#1\right]{\argument}%
}

\def\algPreprocClaPts{\code{preproc\_classical\_pts}}
\def\algPreprocOrdPts{\code{preproc\_ordinary\_pts}}

\def\algComputeHT{\code{compute\_ht}}

\def\cardlw{N^-}
\def\cardord{N^\mathrm{ord}}

\newcommand{\code}[1]{\texttt{#1}}
\def\bitmapccplx{\code{Bitmap\_cubical\_complex}}
\def\EmbeddedCplx{\code{EmbeddedComplex}}
\def\gudhi{\code{GUDHI}}
\def\eucalc{\code{eucalc}}
\def\demeter{\code{demeter}}
\def\sinatra{\code{Sinatra-Pro}}
\def\python{\code{Python}}
\def\cpp{\code{C++}}
\def\fashionmnist{\code{fashion\_MNIST}}

\begin{document}

\maketitle

\begin{abstract}
Topological integral transforms have found many applications in shape analysis, from prediction of clinical outcomes in brain cancer to analysis of barley seeds. Using Euler characteristic as a measure, these objects record rich geometric information on weighted polytopal complexes. While some implementations exist, they only enable discretized representations of the transforms, and they do not handle weighted complexes (such as for instance images). Moreover, recent \emph{hybrid transforms} lack an implementation.

In this paper, we introduce \eucalc{}, a novel implementation of three topological integral transforms---the Euler characteristic transform, the Radon transform, and hybrid transforms---for weighted cubical complexes. Leveraging piecewise linear Morse theory and Euler calculus, the algorithms significantly reduce computational complexity by focusing on critical points. Our software provides exact representations of transforms, handles both binary and grayscale images, and supports multi-core processing. It is publicly available as a \cpp{} library with a \python{} wrapper. We present mathematical foundations, implementation details, and experimental evaluations, demonstrating \eucalc{}'s efficiency. 
\end{abstract}

\section{Introduction}
\label{sec:introduction}

\subsection{Motivations} 
Since its introduction in the late 80s by Viro \cite{V88} and Schapira \cite{S91}, \emph{Euler calculus}---the integral calculus with respect to the Euler characteristic---has been of increasing interest in topological data analysis. Using this integral calculus, one can define integral transforms conveying significant topological information. The most famous one, the \emph{Euler characteristic transform} (ECT), has found many applications in shape analysis, for instance in the prediction of clinical outcomes in brain cancer \cite{CMCMR20}, in the analysis of barley seeds \cite{amezquita2022measuring} or in the recovery of morphological variations across genera of primates \cite{tang2022sinatra}. More formally, given a sufficiently tame compact shape~$K$ in Euclidean space (e.g., an embedded polytopal complex), the ECT records the Euler characteristic of the intersection~$\xi^{-1}(-\infty, t]\cap K$ for all linear forms~$\xi:\R^n\to\R$ and all~$t\in\R$. What is remarkable is that this transform entirely characterizes the shape~$K$, that is, it is injective: if two shapes have same ECT, then they are equal \cite{CMT18,ghrist2018persistent,S95, T14}. The ECT is an instance of a more general transform called the \emph{Radon transform} \cite{S95} and the injectivity of the ECT is a consequence of Theorem~3.1 in loc. cit.. Another natural specialization of the Radon transform, simply called \emph{Radon transform} in this paper, records the Euler characteristic of~$\xi^{-1}(t)\cap K$ for all linear forms~$\xi:\R^n\to\R$ and all~$t\in\R$. 

The most problematic aspect of Euler calculus is its instability under numerical approximations \cite{Cur12}. To remedy it, integral transforms mixing Euler calculus and Lebesgue integration were introduced in~\cite{GR11}, then generalized
under the name \emph{hybrid transforms}~(HT) in~\cite{Lebovici22}. These transforms integrate the Radon transform against a chosen kernel to smooth out and compress the information. The use of a kernel provides a wide variety of potential summaries that can be used to emphasize different aspects of data \cite{Hacquard2023}. All these transforms are defined more generally for any constructible function, but in this article we restrict ourselves to the case of weighted polytopal complexes, i.e., polytopal complexes for which each cell is assigned an integer.

\subsection{Existing implementations and limitations} 
In this paper, we are interested in the computation of topological integral transforms. Several implementations already exist. The \python{} package \demeter{} \cite{demeter,amezquita2022measuring} implements a discretized version of the ECT for axis-aligned cubical complexes. This package turns an $n$-dimensional binary image into an embedded cubical complex~$\ccplx$ in~$\R^n$ in a prescribed way: vertices are defined by non-zero pixels and higher dimensional cubes are defined by adjacency relations. The ECT of~$\ccplx$ is then output for any direction~$\xi:\R^n\to\R$ as a discretization of the \emph{Euler characteristic curve} ~$t\in\R \mapsto \Euler(\xi^{-1}(-\infty, t]\cap \ccplx)$ for some chosen discretization step~\code{T}. The computation of the curve is done using the classical formula for the Euler characteristic~$\Euler(\xi^{-1}(-\infty, t]\cap \ccplx)$ as a sum of terms~$(-1)^{\dim(c)}$ over all cells~$c$ of~$\ccplx$ such that~$c\subseteq \xi^{-1}(-\infty, t]$. A bucket sort of~$\xi$ values on the vertices optimizes the computation of the discretization. The \python{} package \sinatra{} \cite{tang2022sinatra} implements a similar algorithm for triangular meshes both for the ECT and its smoothed version~(SECT)~\cite{CMCMR20}.

These implementations suffer from several essential limitations. 
First of all, they only work with polytopal complexes, not weighted polytopal complexes. This prevents, e.g., the processing of images, which are ubiquitous in applications. Furthermore, these implementations only output discretized Euler characteristic curves. This is important because, while~$2^n$ well-chosen Euler characteristic curves are sufficient to fully reconstruct a cubical complex embedded in~$\R^n$, each curve must be computed exactly for the reconstruction to be correct. Reconstructing the complex from discretized curves will likely lead to substantial errors in the reconstruction, due to the instability of Euler calculus with respect to numerical approximations. On the contrary, an exact computation of integral transforms would pave the way for the implementation of their left-inverses. A software computing an exact version of the ECT on weighted cubical complexes together with an algorithm for its left-inverse does exist, but only for two- and three-dimensional complexes \cite{betthauser}. Besides, this software (as all aforementioned ones) does not allow for the computation of the Radon and hybrid transforms. There is only one recently released software \cite{Hacquard2023} that computes hybrid transforms for some specific complexes arising in the context of multi-parameter persistence. However, no algorithm exists for general polytopal complexes, nor even for cubical ones. To the best of our knowledge, no implementation of the Radon transform exists.

\subsection{Contributions} We introduce fast implementations of the ECT, Radon transform, and hybrid transforms for weighted (axis-aligned) cubical complexes built from grayscale images. After a pre-processing step, our implementation can evaluate the transform exactly at any vector~$\xi$ in time  at most linear (and often largely sublinear in practice) in the number of vertices in the complex. This is thanks to two major optimizations. 

Our first and main optimization comes from the use of piecewise linear Morse theory~\cite{BB97}, which guarantees that, under some assumptions that are automatically satisfied by linear forms~$\xi:\R^n\to \R$ on polytopal complexes~$K$, the topology of~$\xi^{-1}(-\infty, t] \cap K$ changes only when~$t\in\R$ belongs to a finite set of so-called \emph{critical values}. These critical values only appear when~$\xi^{-1}(t)$ contains a vertex of the complex called a \emph{critical point}. The set of critical values of $\xi$ together with the associated changes in the Euler characteristic of sublevel-sets yield a finite and exact representation of the associated Euler characteristic curve---see for instance \cite{CMT18} for a precise statement in the case of unweighted cubical complexes. This reduces the computation from a sum over all cells to a sum over only a few critical points. As an illustration, we show that there are between 34 and 21 times fewer critical points than vertices on the images of the \fashionmnist{} data set.

Our second optimization comes from the observation that two linear forms~$\xi$ and~$\xi'$ whose canonical coordinates have same signs share the same critical points on axis-aligned cubical complexes. Precomputing these points for~$2^n$ well-chosen directions enables fast evaluation of topological transforms in many directions at a very low cost. 
Moreover, the formalism of Euler calculus suggests that the results of piecewise linear Morse theory also hold for weighted cubical complexes. We exploit this remark in our implementation. Consequently, computing topological transforms of grayscale images is done in less than twice the time for binary ones.  

Our software, called \eucalc{}~\cite{eucalc}, is available as a \cpp{} library together with a \python{} wrapper to make it easily compatible with scikit-learn like interfaces. It provides a data structure for weighted embedded cubical complexes allowing for fast preprocessing of critical points. Our data structure builds on the implementation of cubical complexes from the \gudhi{} library to benefit from their flexibility of instanciation. Moreover, our structure represents exact versions of the ECT and the Radon transform, together with evaluation and discretization routines. Our implementation also includes a parallelized version of the preprocessing step and the computation of the hybrid transform.

\section{Preliminaries}\label{sec:preliminaries}
Let~$n\geq 2$ be an integer. The space of linear forms on~$\R^n$ will be canonically identified with~$\R^n$. For clarity, we refer to linear forms with greek letters~$\xi\in\R^n$, and we denote the canonical dot product of~$\xi$ and~$x\in\R^n$ by~$\dualdot{\xi}{x}$.
For any function~$f:\R^n\to\R$ and~$t\in\R$, we denote by~$\{f = t\}$ and~$\{f \leq t\}$ the level set~$\{x\in\R^n \st f(x) = t\}$ and the sublevel set~$\{x\in\R^n \st f(x) \leq t\}$ respectively.
\subsection{Weighted polytopal complexes}
While our algorithms deal with cubical complexes, we introduce our main objects in the more general setting of polytopal complexes for the sake of clarity. We refer to Ziegler \cite{Z06} for more details on polytopal complexes. 
We call \emph{polytope} a bounded intersection of a finite number of closed half-spaces of~$\R^n$. Equivalently, a polytope is the convex hull of a finite number of points in~$\R^n$, called its \emph{vertices} \cite[Thm.~1.1]{Z06}. We denote by~$\vert{P}$ the (finite) set of all vertices of a polytope~$P$, and we call it its \emph{vertex set}. 
A \emph{face} of~$P$ is any set of the form~$F = P \cap \{\xi = t\}$ where~$\xi\in\R^n$ and~$t\in \R$ are such that~$P\subseteq \{\xi\leq t\}$. 
The \emph{dimension} of a polytope is the dimension of the smallest affine subspace of~$\R^n$ containing it.
A \emph{polytopal complex}~$\cplx$ is a finite collection of polytopes in~$\R^n$ (called \emph{cells}) such that (i) the empty cell is in~$\cplx$, (ii) if~$P$ is in~$\cplx$, then all the faces of~$P$ are also in~$\cplx$, and (iii) if~$P$ and~$Q$ are in~$\cplx$, then their intersection is a face of both~$P$ and~$Q$. One can always extend a collection of polytopes satisfying (i) and (iii) to a polytopal complex by adding all faces of the polytopes in the collection. This polytopal complex is the smallest one containing the polytopes of the given collection.

In this paper, we are mainly interested in cubical complexes. A \emph{cube} in~$\R^n$ is a subset~$\cube = [a_1,b_1]\times\ldots\times[a_n,b_n]\subseteq \R^n$ such that~$b_i-a_i = b_j - a_j$ if~$a_i \ne b_i$ and~$a_j \ne b_j$. A \emph{cubical complex}~$\cplx$ is a polytopal complex whose polytopes are cubes in~$\R^n$. 

Let~$\cplx$ be a polytopal complex in~$\R^n$ and~$\xi\in \R^n$. The \emph{star} of a polytope~$P\in\cplx$ is the collection~$\str{P}$ of polytopes~$Q\in\cplx$ such that~$Q\supseteq P$. For any~$x\in\vert{\cplx}$, the \emph{lower star} of~$x$ with respect to~$\xi$ is the collection~$\lwst{\xi,x}$ of polytopes~$\cubee \in\str{x}$ such that~$\max_\cubee \xi = \xi(x)$. Similarly, we define the \emph{upper star}~$\upst{\xi,x}$. These collections of polytopes are not complexes as they are not closed under taking faces.

Let~$\collec$ be a finite collection of polytopes in~$\R^n$ and let~$\phi: \collec \to \Z$. If~$\collec$ is a polytopal complex, then we call~$\phi$ a \emph{weighted polytopal complex}. The \emph{Euler characteristic} of~$\phi$ is defined as~$\Euler(\phi) = \sum_{P\in\collec} (-1)^{\dim(P)} \phi(P)$. This definition generalizes the usual notion of Euler characteristic of a polytopal complex \cite[Corollary~8.17]{Z06} and is a particular case of Euler integration~\cite{V88,S91}. 

\subsection{Integral transforms} 
Let~$\xi\in\R^n$ and~$t\in\R$. Consider the collections of polytopes~$\{P\cap\{\xi\leq t\} : P\in\cplx\}$ and~$\{P\cap\{\xi=t\} : P\in\cplx\}$ and denote respectively by~$\cplxleq$ and~$\cplxeq$ the smallest polytopal complexes containing these collections. Any function~$\phi:\cplx\to\Z$ induces a function~$\phileq:\cplxleq\to \Z$ defined for any~$Q\in\cplxleq$ by~$\phileq(Q) = \phi(P)$ where~$P$ is the smallest polytope of~$\cplx$ containing~$Q$. In that case, one has $\dim(Q)=\dim(P)-1$. Similarly, the function~$\phi$ induces a function~$\phi_{\{\xi=t\}}:\cplxeq\to \Z$.

Let~$\phi:\cplx\to\Z$. The first topological integral transform we consider in this paper is the \emph{Radon transform}~\cite{S95}, namely, the function~$\Radon:\R^n\times\R \to \Z$ defined for any~$(\xi, t)\in\R^n\times\R$ by:
\begin{equation*}
	\Radon(\xi, t) = \Euler[\phieq].
\end{equation*} 
The second integral transform, the \emph{Euler characteristic transform}~\cite{T14}, is the function~$\ECT:\R^n\times\R \to \Z$ defined for any~$(\xi, t)\in\R^n\times\R$ by:
\begin{equation*}
	\ECT(\xi, t) = \Euler[\phileq].
\end{equation*} 
Following \cite{Lebovici22}, we consider specific smoothings of topological transforms. Let~$\kernel:\R\to\C$ be locally integrable. The \emph{hybrid transform} of~$\phi$ with kernel~$\kernel$ is the function~$\HT:\R^n \to \C$ defined for any~$\xi\in\R^n$ by:
\begin{equation*}
	\HT(\xi) = \int_\R \kernel(t) \Radon(\xi, t) \dt.
\end{equation*} 

\subsection{Naive algorithm}
The following lemma is an easy consequence of the definitions.
\begin{lemma}\label{lem:naive}
    Let~$\bkernel:\R\to\C$ be a primitive of~$\kernel$.
	\begin{align*}
		\Radon(\xi,\cdot) &= \sum_{P\in\cplx} \phi(P)\,\1_{[\min_P\xi,\max_P\xi]},\\
		\ECT(\xi,\cdot) &= \sum_{P\in\cplx} \phi(P)\,\1_{[\min_P\xi,+\infty)}.
	\end{align*}
	As a consequence, 
	\begin{equation*}
		\HT(\xi) = \sum_{P\in\cplx} \phi(P) \cdot \left(\bkernel\left(\max_P \xi\right) - \bkernel\left(\min_P \xi\right)\right).
	\end{equation*}
\end{lemma}
This lemma straightforwardly suggests naive algorithms to compute topological integral transforms over any polytopal complex. Doing so, the time complexity of computing any one of the above integral transforms over~$\cplx$ is~$\bigO(N \cdot \card \cplx)$ where $N$ is the maximum number of vertices of cells in $\cplx$. These naive algorithms will serve as a baseline for our future optimizations. Notice that $\card \cplx$ is exponential in the dimension.

\section{Euler critical points}\label{sec:Euler-critical-points}
In this section, we slightly generalize the notion of Euler critical value of polytopal complexes \cite[Def.~6.2]{CMT18} to weighted polytopal complexes. All mathematical results are folklore. We state them and provide proofs in \Cref{app:proofs} for completeness. 

Let~$\cplx$ be a polytopal complex,~$\xi\in\R^n$ and~$\phi:\cplx\to\Z$. Further assume that~$\xi$ is \emph{generic} on~$\cplx$, i.e., no polytope of~$\cplx$ of positive dimension is contained in a level set of~$\xi$. We define the \emph{ordinary} and \emph{classical critical values} of~$\xi$ at a vertex~$x$ of $\cplx$ respectively~by:
\begin{align*}
    \ordcritvar{\xi,x} &= \Euler[\phi_{|\upst{\xi,x}}] - \Euler[\phi_{|\lwst{\xi,x}}], \\[0.5em]
    \lwcritvar{\xi,x} &= \Euler[\phi_{|\lwst{\xi,x}}].
\end{align*}
A vertex~$x$ of~$\cplx$ is an \emph{Euler~$\phi$-critical point} of~$\xi$ if one of~$\lwcritvar{\xi,x}$ or~$\ordcritvar{\xi,x}$ is non-zero. Such a point is called \emph{ordinary} when~$\ordcritvar{\xi,x} \ne 0$ and \emph{classical} when~$\lwcritvar{\xi,x}\ne 0$. Euler~$\phi$-critical points are sometimes called \emph{critical points} when~$\phi$ and~$\xi$ are clear from the context. Note that a critical point can be both classical and ordinary. The set of classical critical points is denoted by~$\lwcritpts{\xi}$, the set of ordinary critical points by~$\ordcritpts{\xi}$ and the set of all critical points by~$\critpts{\xi} = \lwcritpts{\xi} \cup \ordcritpts{\xi}$. Critical points and values are illustrated in \Cref{app:running-example}. 
The following lemma expresses integral transforms as sums over critical points. It is key to our algorithms.

\begin{lemma}\label{lem:expression-transforms-crit}
    Let~$\phi:\cplx\to\Z$ and let~$\xi\in\R^n$ be generic on~$\cplx$. Then,
    \begin{align*}
        \Radon(\xi, \cdot) &= \sum_{x\in\lwcritpts{\xi}} \lwcritvar{\xi,x} \1_{\{\dualdot{\xi}{x}\}}  + \sum_{x\in\ordcritpts{\xi}} \ordcritvar{\xi,x}\1_{(\dualdot{\xi}{x},+\infty)}, \\[0.2em]
        \ECT(\xi, \cdot) &= \sum_{x\in\lwcritpts{\xi}} \lwcritvar{\xi,x} \1_{[\dualdot{\xi}{x},+\infty)}.
    \end{align*}
	Furthermore, if~$\kernel:\R\to\C$ is locally integrable and~$\bkernel$ is a primitive of~$\kernel$ on~$\R$, one has:
	\begin{equation*}
		\HT(\xi) = - \hspace{-1.1em} \sum_{x\in\ordcritpts{\xi}} \ordcritvar{\xi,x} \cdot \bkernel\big(\dualdot{\xi}{x}\big).
	\end{equation*}
\end{lemma}

Let~$\ccplx$ be a cubical complex. The result below can be adapted to arbitrary polytopal complexes but we do not make use of such a general statement in our algorithms. Since our definition of cubes only allow for axis-aligned cubes, any~$\xi\in\R^n$ with non-zero coordinates is generic on~$\ccplx$. 
For any~$\xi=(\xi_1,\ldots,\xi_n)\in\R^n$ with non-zero coordinates we define its \emph{sign vector}~$\sgnvec{\xi} \in \{\pm 1\}^n$ by~$\sgnvec{\xi}_i = \sgn(\xi_i)$. The following lemma ensures that only~$2^n$ sets of critical points and values are necessary to compute integral transforms in practice. 
\begin{lemma}\label{lem:hyperplane-arrangement}
	Let~$\xi,\xi'\in\big(\R\setminus\{0\}\big)^n$. If~$\sgnvec{\xi}=\sgnvec{\xi'}$, then
    for any~$\phi:\ccplx\to\Z$ and any~$x\in\vert{\ccplx}$ the sets and numbers~$\lwst{\xi,x}$,~$\upst{\xi,x}$,~$\lwcritvar{\xi, x}$,~$\ordcritvar{\xi, x}$, $\lwcritpts{\xi}$ and~$\ordcritpts{\xi}$ are respectively equal to the ones associated to~$\xi'$.
\end{lemma}

\section{Implementation and asymptotic analysis}\label{sec:implementation}

The mathematical results of the previous section suggest algorithms allowing for efficient computation of the three integral transforms under consideration in this paper. In this section, we describe our implementation and provide worst-case time complexity bounds.
\subsection{Implementation of cubical complexes}
Our implementation of cubical complexes consists in a \cpp{} class \EmbeddedCplx{} that inherits from the cubical complex class \bitmapccplx{} from the \gudhi{} library (release 3.8.0) and additionally stores the embedding coordinates of the vertices in~$\R^n$.  This inheritance allows us to take any $n$-dimensional array/image as input and use \gudhi{}'s initialization to turn it into a cubical complex. The \gudhi{} library stores the weighted cubical complex as a vector of bits representing the value assigned to each cell, with a specific convention for the indexing of cells. By default, one top-dimensional cell is associated to each pixel and the values of lower-dimensional cells are equal to the minimum of the values of their cofaces. By convention, the embedding coordinates of the vertices in $\R^n$ are evenly spaced and normalized in~$[-0.5,0.5]^n$.  Storing the values of the complex in \gudhi{}'s initialization is linear in the number of cells in the complex. We do not consider it as part of our algorithm. However, we take it into account in the running times in our experiments of~\Cref{sec:experiments}.

\subsection{Preprocessing of critical points and values}
After the initialization, the first step to compute integral transforms using \Cref{lem:expression-transforms-crit} is to compute the critical values~$\lwcritvar{\eps, x}$ and~$\ordcritvar{\eps, x}$ for all~$x\in\vert{\ccplx}$ and all sign vectors~$\eps\in\{\pm 1\}^n$. While doing so, we  store on the fly each one of the families~$\lwcritvar{\eps}:=\{\lwcritvar{\eps, x}\}_{x\in\lwcritpts{\eps}}$ and~$\ordcritvar{\eps} := \{\ordcritvar{\eps, x}\}_{x\in\ordcritpts{\eps}}$ as a pair of two arrays, one for the critical points and one for the associated critical values, which induces no overhead in terms of asymptotic running time. Using \gudhi{}'s implementation of weighted cubical complexes, accessing the dimension~$\dim(c)$ and the value~$\phi(c)$ of a cell $c$ takes constant time. Moreover, cells are indexed by a single number in this implementation. As a consequence, iterating through the $\bigO(2^n)$ adjacent cells of a vertex requires arithmetic operations to compute indices of adjacent cells, resulting in $\bigO(n2^n)$ time. The computation of critical values at a vertex thus has the same complexity, and the overall time complexity of computing critical values for all vertices and all sign vectors is then~$\bigO(n4^n\cdot\card{\vert{\ccplx}})$.

\subsection{Exact representation and evaluation} 
After the preprocessing step, the computation of the transforms using \Cref{lem:expression-transforms-crit} goes as follows. Let~$\xi\in\R^n$, and denote~$\card{\lwcritpts{\eps^{\xi}}}$ and~$\card{\ordcritpts{\eps^{\xi}}}$ respectively by~$\cardlw$ and~$\cardord$. First, we compute the sign vector~$\sgnvec{\xi}$ in~$\bigO(n)$ time, which will always be negligible. 

To compute the hybrid transform~$\HT(\xi)$, we sum, over all ordinary critical points $x\in \ordcritpts{\eps^{\xi}}$, the evaluation---assumed to be done in constant time---of the kernel~$\bkernel$ on the dot products $\dualdot{\xi}{x}$. Hence a running time in~$\bigO(n\cardord)$.

To compute the Euler characteristic transform~$\ECT(\xi,\cdot)$, we first compute $T=\{\dualdot{\xi}{x}\}_{x \in \lwcritpts{\sgnvec{\xi}}}$ in $\bigO(n\cardlw)$ time. Then, we sort the set $T$ and the sets~$\lwcritpts{\sgnvec{\xi}}$ and~$\lwcritvar{\sgnvec{\xi}}$ by ascending values of~$\dualdot{\xi}{x}$ in~$\bigO\left(\cardlw\log\cardlw\right)$ time. Finally, we compute a vector~$E$ that contains the cumulative sums of the sorted values~$\lwcritvar{\sgnvec{\xi}}$ in~$\bigO(\cardlw)$ time. An exact representation of the function~$\ECT(\xi,\cdot)$ is then returned as the pair of vectors $(T,E)$. The total running time is in~$\bigO\left(\cardlw(n+ \log\cardlw)\right)$.

To compute the Radon transform, we do the same as for the ECT but replacing~$\lwcritpts{\sgnvec{\xi}}$ by~$\ordcritpts{\sgnvec{\xi}}$ and~$\lwcritvar{\sgnvec{\xi}}$ by~$\ordcritvar{\sgnvec{\xi}}$, to get an analogously constructed pair $(T,E')$ in~$\bigO\left(\cardord(n+ \log\cardord)\right)$ time. Then, we sort the elements~$x\in \lwcritpts{\sgnvec{\xi}}$ by ascending values of~$\dualdot{\xi}{x}$ in~$\bigO\left(\cardlw(n+ \log\cardlw)\right)$ time. After this, we use binary search to find, for each~$x\in\lwcritpts{\sgnvec{\xi}}$, the critical point~$y_x\in\ordcritpts{\sgnvec{\xi}}$ with highest value~$\dualdot{\xi}{y_x} < \dualdot{\xi}{x}$. This takes $\bigO(\cardlw\log\cardord)$ time. We store~$E'[y_x]+\lwcritvar{\sgnvec{\xi},x}$ for each~$x\in\lwcritpts{\sgnvec{\xi}}$ in an array~$E_c$. If there exist several points~$x' \in \lwcritpts{\xi}$ with the same dot product~$\dualdot{\xi}{x'}$, we only store the cumulative sum of the critical values~$\lwcritvar{\sgnvec{\xi},x'}$ in one entry of~$E_c$. An exact representation of the function~$\Radon(\xi,\cdot)$ is then returned as the tuple of the four sorted vectors~$T$,~$E'$,~$T_c = \{\dualdot{\xi}{x}\}_{x \in \lwcritpts{\sgnvec{\xi}}}$ and~$E_c$, in~$\bigO(\cardord \cdot n + (\cardord + \cardlw)\cdot \log\cardord)$ time. 

\begin{remark}
    To better understand these complexities, note that the number of vertices of~$\ccplx$ is asymptotically equivalent to the number of pixels in the image. However, the total number of vertices appears only in the complexity of the preprocessing step.
    In many applications, the numbers~$\cardord$ and~$\cardlw$ are small compared to~$\card{\vert{\ccplx}}$ so that the computation step is significantly optimized. This heuristic will be confirmed in \Cref{sec:experiments}.
\end{remark}

To evaluate the Euler characteristic transform~$\ECT(\xi,\cdot)$ at a point $t\in\R$ using its exact representation~$(T,E)$, we use a binary search to find the element~$x \in \ordcritpts{\eps}$ with greatest~$\dualdot{\xi}{x} \leq t$ and return~$E[x]$ in $\bigO(\log\cardord)$. We also implemented a vectorization routine to get the values of the ECT at sorted and evenly spaced points inside a specified interval. This is done by iterating simultaneously over $T$ and over the points of the interval, hence a running time in~$\mathcal{O}(N+\cardord)$. 

The algorithm to evaluate the Radon transform using its exact representation $(T,E',T_c,E_c)$ is identical, except when there exists~$x \in \lwcritpts{\eps}$ such that~$t = \dualdot{\xi}{x}$, in which case it returns~$E_c[x]$. Again, we use binary search to determine the existence of such a point~$x$ in~$\bigO(\log \cardlw)$ time. 

\section{Experimental evaluation}\label{sec:experiments}
In this section, we illustrate our optimizations on several batches of experiments. In all experiments, integral transforms are computed in 100 directions drawn uniformly at random in the cube~$[0,1]^n$. We compute hybrid transforms with $\kernel(t) = \sin(t)$. To allow for a comparison with \demeter{}, we instantiate our cubical complexes using the \emph{dual construction}, that is, initializing the values of the vertices with the pixel values in the images, and setting the values of higher-dimensional cells as the minimum of the values of their vertices. Moreover, we compare our exact version of the ECT to the discretized version of \demeter{} with resolution \code{T=100}. This is justified by the fact that our vectorization step turning our exact version to a discretized one has negligible cost compared to the other steps, as shown in \Cref{fig:vect-by-pts}. 

Our code has been compiled with gcc~9.4.0 and run on a workstation with an Intel(R) Core(TM) i7-4770 CPU (8 cores, 3.4GHz) and 8 GB RAM, running GNU/Linux (Ubuntu 20.04.1). We measure time consumption using the \python{} module \code{time}. The input data and the benchmark \python{} scripts are available upon request.
\subsection{Data set example}\label{sec:fashion-mnist}
To begin with this experimental section, we compute integral transforms on a real-world data set: \fashionmnist{}. This data set contains 60000 grayscale images of size $28\times 28$ representing several types of clothes. We chose this data set because it contained global geometric information, namely the shape of clothes.

Results are shown in \Cref{tab:fashion-mnist}. In the first three lines, we timed our implementation of integral transforms on the \fashionmnist{} data set with unchanged pixel values, potentially ranging from 0 to 255. In the next four lines, we timed our implementation as well as the ECT of \demeter{} on a binarized version of the \fashionmnist{} data set, where the 50\% brightest pixels of each image are set to 1 and all other pixels to 0. This binarization was necessary since \demeter{} does not cope with grayscale images. 

\begin{table}[H]
    \centering
    \begin{tabular}{|c||c|c|c|c|}
    \hline
       Software & Init. & Preproc. & Comp. & Total \\ \hline\hline 
        Radon     &  13 &  143 &  290 & 447 \\ \hline
        HT        & 13  & 91 &  32 & 137 \\ \hline
        ECT       &  13 &  54 & 117  & 184 \\ \hline\hline
        Radon (binary)     & 13 & 138 & 56 & 207 \\ \hline
        HT (binary)        & 13 & 88 & 14 & 115 \\ \hline
        ECT (binary)     & 13 & 52 & 32 & 97 \\ \hline
        \demeter{} ECT  & 0 & 266 & 1544 & 1810 \\
        \hline
    \end{tabular}
    \caption{Timing (s) of our implementation and of the \demeter{} ECT on the \fashionmnist{} data set.\\ \vspace{-2em}}
    \label{tab:fashion-mnist}
\end{table}

What stands out from this experiment is that all transforms can be computed in a reasonable amount of time on this large data set. The preprocessing step seems to be even more efficient in the case of hybrid transforms, where the evaluation step in 100 directions is done in between a third and a sixth of the time of the preprocessing step.

More importantly, this experiment perfectly illustrates the efficiency of our implementation. Each $28\times 28$ image of the \fashionmnist{} data set induces a cubical complex with 784 vertices. On average, there are~$172$ classical critical points with a standard deviation of~$50$ and~$237$ ordinary critical points with a standard deviation of~$65$. Thus, around a quarter of the vertices are classical critical points, and the same is true for ordinary critical points. In contrast, the complexes of the binarized data set have $23$ classical critical points on average with a standard deviation of $12$ and $36$ ordinary critical points with a standard deviation of $19$. As a consequence, after the initialization and preprocessing steps, our optimization using critical points allows us to compute an exact version of the ECT on the binarized version more than 50 times faster than \demeter{} computes its discretized version. Overall, our implementation is 18 times faster than \demeter{}.

Moreover, thanks to our generalization of Euler critical values to weighted polytopal complexes, it is still 10 times faster to compute our exact version of the ECT on the original grayscale data set than to compute it on the binarized data set using \demeter{}.

\subsection{Asymptotic analysis}\label{sec:asymptotical-analysis}
In this section, we perform an experimental asymptotic analysis of our algorithms. This analysis confirms our theoretical results on time complexities of \Cref{sec:implementation}. 

\paragraph*{Size of images} We study the computation times on~$m\times m$ images with respect to the parameter~$m$. We use binary images with constant number of critical points made of a periodic pattern of black squares in a white background. The number of squares is constant when~$m$ varies to ensure a constant number of critical points. Namely, they are~$100$ classical critical points and~$200$ ordinary critical points for all images.

As one could expect, the initialization (\Cref{fig:size-init}) seems linear in the number of cells in the complex, that is, quadratic in~$m$. Similarly, the preprocessing step (\Cref{fig:size-preproc}) is linear in the number of vertices in the complex, that is, quadratic in~$m$. The strength of our implementation is illustrated by \Cref{fig:size-eval-dem}. After the initialization and the preprocessing steps, computing the integral transforms is done in constant time when the number of critical points is fixed. The curves for our times in \Cref{fig:size-eval-dem} are constant and lie between 1 and 5 milliseconds.

Overall, our computation times are dominated by the preprocessing step. Remarkably, our implementation has improved computation times by a factor of 4 compared with \demeter{}.
\begin{figure}
    \centering
    \begin{subfigure}[t]{0.32\linewidth}
        \centering
        \includegraphics[width=0.93\linewidth]{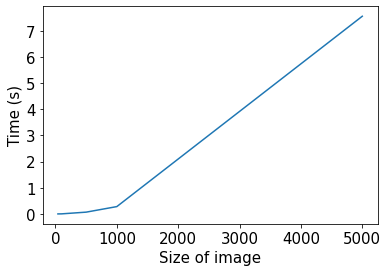}
        \caption{Initialization}
        \label{fig:size-init}
    \end{subfigure}
    \begin{subfigure}[t]{0.32\linewidth}
        \centering
        \includegraphics[width=0.93\linewidth]{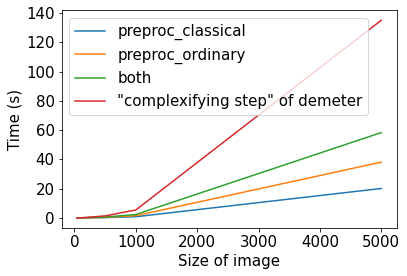}
        \caption{Preprocessing or ``complexification'' of \demeter}
        \label{fig:size-preproc}
    \end{subfigure}
    \begin{subfigure}[t]{0.32\linewidth}
        \centering
        \includegraphics[width=0.93\linewidth]{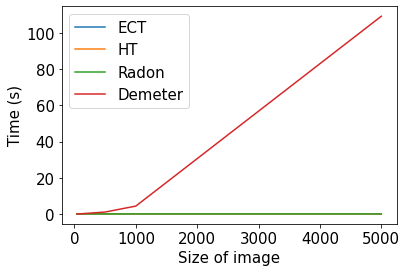}
        \caption{Computation times excluding initialization and preprocessing}
        \label{fig:size-eval-dem}
    \end{subfigure}
    \caption{Step-by-step computation times for varying size~$m$ of~$m\times m$ images.}
\end{figure}

\paragraph*{Number of critical points}
We study the computation times with respect to the number of critical points in~$1000\times 1000$ images. We use the same periodic pattern of squares to control the number of critical points, this time with a varying number of squares.

The initialization and preprocessing steps are done in constant time for a constant size of image. More precisely the initialization takes around 1 second and the preprocessing of classical and ordinary critical points take approximately 3 and 6 seconds respectively. The ``complexifying'' step of \demeter{} also takes around 6 seconds. 

After the initialization and the preprocessing steps, computing the integral transforms is done in linear (HT) and linearithmic (ECT, Radon) times in the number of critical points for a fixed size of images. This is clearly illustrated in \Cref{fig:crit-pts-eval-dem}. Overall, computing an Euler characteristic curve for a binary image with four millions pixels takes between a 0.1 milliseconds to a few seconds depending on the number of critical points.

Interestingly, the computation time for the \demeter{} package decreases with the number of critical points. This is due to our periodic patterns with varying size used to generate images with a varying number of critical points. The smaller the patterns are, the greater the number of critical points, but the lesser edges and two-dimensional cubes in the complex. As a consequence, when the number of critical points increases up to one quarter of the total number of vertices, \demeter{} has a faster evaluation step than our evaluation of the ECT (\Cref{fig:crit-pts-eval-dem}). Time for the total pipeline are thus very similar between their and our implementation (\Cref{fig:crit-pts-total-dem}). 

However, this fact can be mitigated by two remarks. First of all, we compute an exact version of the Euler characteristic transforms and not an approximated version with a resolution of \code{T=100}. This approximated version is far from the exact one, as they are 25000 changes of values in the true Euler characteristic curve when they are 25000 classical critical points. Secondly, the number of critical points in real-world binary images containing global geometric features are likely to be way lower than a quarter of the vertices, as suggested by our study on a real-world data set in \Cref{sec:fashion-mnist}. While the number of critical points is likely to be higher for grayscale images, the \demeter{} package can no longer be used in this case.

\begin{figure}
    \centering
    \begin{subfigure}[t]{0.49\linewidth}
        \centering
        \includegraphics[width=\linewidth]{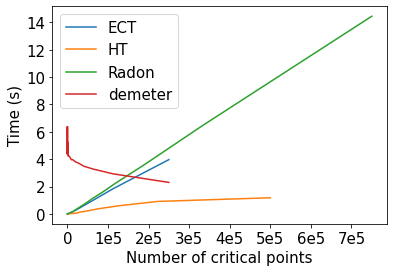}
        \caption{}
        \label{fig:crit-pts-eval-dem}
    \end{subfigure}
    \begin{subfigure}[t]{0.49\linewidth}
        \centering
        \includegraphics[width=\linewidth]{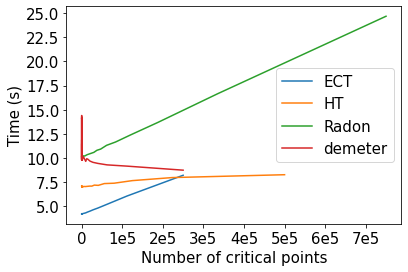}
        \caption{}
        \label{fig:crit-pts-total-dem}
    \end{subfigure}
    \caption{Computation times for varying number of critical points excluding (\ref{fig:crit-pts-eval-dem}) and including (\ref{fig:crit-pts-total-dem}) initialization and preprocessing steps.}
\end{figure}

\paragraph*{Number of directions}
We study the influence of the number of requested directions on the total running time. All computations are done on a random binarized image of the fashion MNIST dataset. We run the computation of the ECT (including initialization and preprocessing) for 1 to 10000 directions sampled uniformly in~$[0,1]^2$, and measure the corresponding times. The initial value of each curve corresponds to the initialization time. We give these plots on log-scale in \Cref{fig:fashion-mnist-log}, then on linear scale with the initial values forced to zero (which corresponds to ignoring the initialization time) in \Cref{fig:fashio-mnist-linear}. By a linear regression, we check that the points we describe sample a linear function and we determine that its slope is $2.8 \cdot 10^{-6}$ seconds by direction for our method, and $2.3 \cdot 10^{-4}$ seconds by direction for \demeter{}.

\begin{figure}
    \centering
    \begin{subfigure}[t]{0.45\linewidth}
        \centering
        \includegraphics[width=\linewidth]{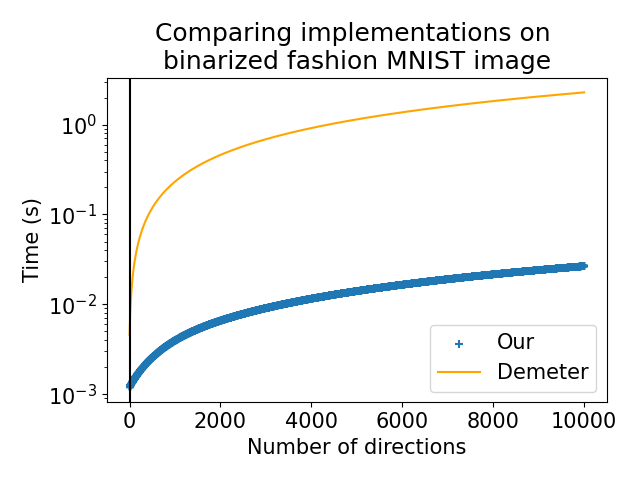}
        \caption{Logarithmic scale}
        \label{fig:fashion-mnist-log}
    \end{subfigure}
    \begin{subfigure}[t]{0.45\linewidth}
        \centering
        \includegraphics[width=\linewidth]{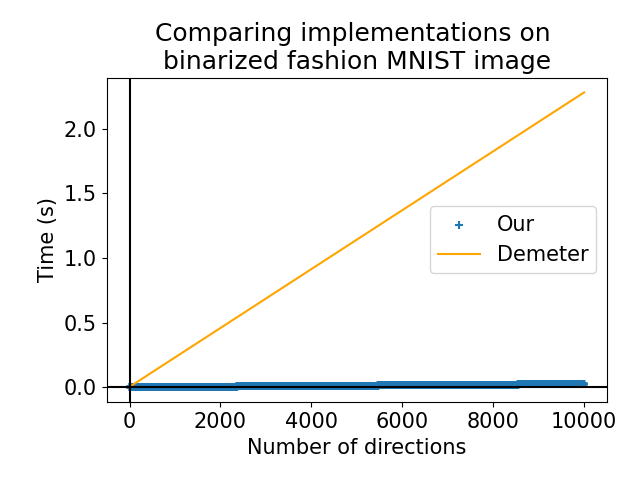}
        \caption{Linear scale}
        \label{fig:fashio-mnist-linear}
    \end{subfigure}
    \caption{Computation times (s) for varying number of directions of our and \demeter{}'s implementation of the ECT on a single image of the fashion-MNIST dataset.}
    \label{fig:fashion-mnist-comparison}
\end{figure}

\paragraph*{Vectorization} 
On \Cref{fig:vect-by-pts} we ran the vectorization routine on 1000x1000 image with $10^6$ critical points for a varying dimension of vectorization, that is, sampling the ECT on an increasing number of points $t\in\R$. The complexity of the vectorization routine is confirmed by \Cref{fig:vect-by-pts}. More importantly, vectorization is negligible (a few milliseconds) compared to other steps of the computation of the transform (a few seconds each). 
\begin{figure}[!h]
    \centering\includegraphics[height=0.17\paperheight]{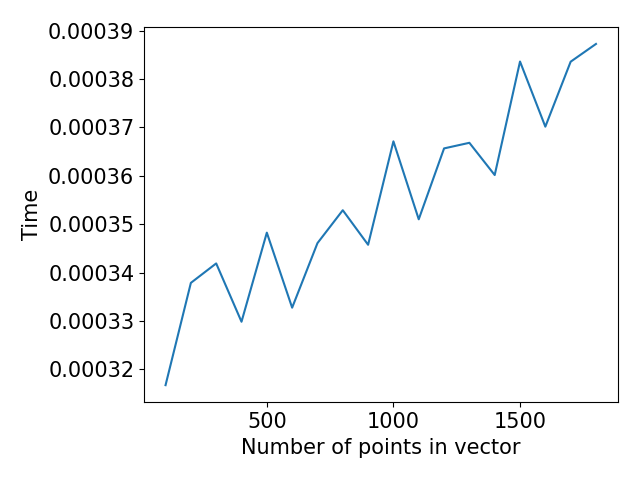}
    \caption{Vectorization time (s) with respect to dimension for a $1000\times 1000$ image with~$10^6$ critical points.\vspace{-1.5em}}
    \label{fig:vect-by-pts}
\end{figure}

\paragraph*{Average number of critical points}
In this section, we expose an experimental study of the average number of critical points in random settings for two-dimensional grayscale images as the number of shades of gray increases. Pixel values are drawn using a uniform distribution in a first setting and using a Gaussian random field in a second setting.

In both cases, the number of both types of critical points in grayscale images is approximately half the number of vertices as soon as the number of pixel values is above 10. In the binary case, the number of critical points of each type is between 10\% and 20\% of the total number of vertices. These results corroborate  the computation times observed on the data sets presented above: the number of critical points is in practice lower than the number of vertices, and much lower than the number of cells in the complex. 

While this experiment shows that the number of critical point is of the same order as the number of vertices, it is expected not to be the case on digital images with global geometric features, as suggested in \Cref{sec:fashion-mnist}.

\begin{figure}[!h]
    \centering
\includegraphics[height=0.17\paperheight]{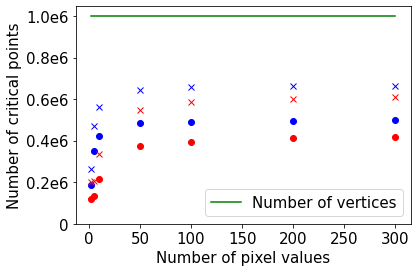}
    \caption{Number of critical points in two random settings. Red and blue symbols represent respectively classical and ordinary critical points. Dots and crosses respectively represent these numbers for the uniform setting and for the Gaussian random field setting.\vspace{-1.5em}}
    \label{fig:random-crit-pts}
\end{figure}

\begin{figure}[!h]
    \centering
    \includegraphics[height=0.17\paperheight]{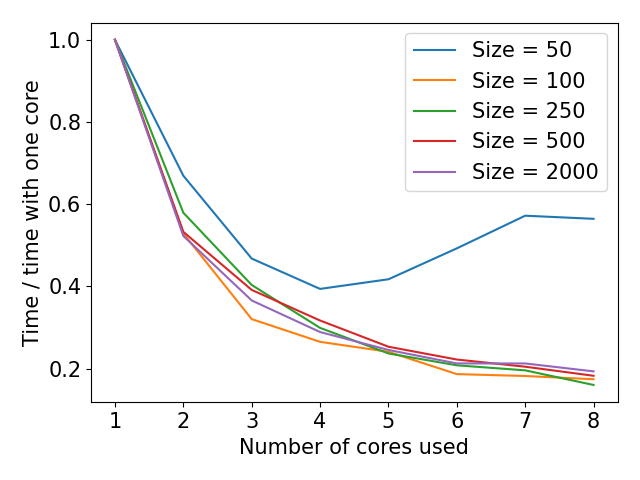}
    \caption{Relative pre-processing time (s) given the number of cores used for several complex sizes.\vspace{-1.5em}}
    \label{fig:parallelization}
\end{figure}

\subsection{Parallelization}
We parallelized the routines \algPreprocOrdPts{} and \algPreprocClaPts{} as they are the most time-consuming functions. To do so, we parallelized the inner loop on~$\vert{\ccplx}$, each core handling the computations on a subset of the vertices. The readings are concurrent as cores only need to access fixed cell values. Each core returns an array of critical points and an array of critical values and the main core concatenates the results. Parallelizing the inner loop instead of the main one is efficient because the number of sign vectors in dimension~$n$ is~$2^n$, and $n=2$ or $3$ in most practical cases.
We also implemented a parallelization of the routine \algComputeHT{} that takes as input a vector of directions and returns a vector containing the values of the hybrid transfoms in these directions. In this case each core handles the computations on a subset of the directions and returns the result to the main core that concatenates them. On \Cref{fig:parallelization} we study the time of the preprocessing of critical points (both classical and ordinary) given the number of cores used, relatively to the single core setup. We do this for several complex sizes with cells values uniformly sampled in $\{0,1\}$. We see that the parallelization induces a really small overhead, even for small complexes. Moreover, recall that even with only one core, our optimizations allow our software to be 100 times faster than the \demeter{} package on images of the \fashionmnist{} data set.

\section{Conclusion}\label{sec:conclusion}

We presented a novel implementation of topological integral transforms on weighted cubical complexes, based on existing results from piecewise linear Morse theory slightly adapted to weighted polytopal complexes. These results allow for the fast computation of integral transforms on weighted cubical complexes. As a consequence, our implementation can efficiently compute exact versions of the ECT and Radon transform, and we showed experimentally that it is 18 times faster than the \demeter{} ECT on a real-world data set, and 50 times faster after the preprocessing step. We also added a vectorization routine for data analysis purposes.

A natural extension of our implementation would be to arbitrary polytopal complexes. For any such complex, there is a cellular decomposition of (the dual of)~$\R^n$ such that \Cref{lem:hyperplane-arrangement} holds for any $\xi, \xi'$ lying in the same cell~\cite{CMT18}. This cellular decomposition is induced by the arrangement of hyperplanes that are orthogonal to the edges of the complex. The main challenge is that the size of this arrangement for a general polytopal complex with $n$ vertices in $\R^d$ may be as large as $n^{2d}$, as opposed to only $2^d$ for cubical complexes.



\bibliography{mybib}

\begin{thebibliography}{10}

\bibitem{demeter}
Erik~J Am{\'e}zquita.
\newblock \texttt{demeter}.
\newblock \url{https://github.com/amezqui3/demeter}, 2021.

\bibitem{amezquita2022measuring}
Erik~J Am{\'e}zquita, Michelle~Y Quigley, Tim Ophelders, Jacob~B Landis, Daniel
  Koenig, Elizabeth Munch, and Daniel~H Chitwood.
\newblock Measuring hidden phenotype: Quantifying the shape of barley seeds
  using the euler characteristic transform.
\newblock {\em in silico Plants}, 4(1), 2022.

\bibitem{BB97}
Mladen Bestvina and Noel Brady.
\newblock Morse theory and finiteness properties of groups.
\newblock {\em Inventiones mathematicae}, 129(3):445--470, 1997.

\bibitem{betthauser}
Leo Betthauser.
\newblock {Topological-Reconstruction}.
\newblock \url{https://github.com/lbetthauser/Topological-Reconstruction},
  2018.

\bibitem{CMCMR20}
Lorin Crawford, Anthea Monod, Andrew~X. Chen, Sayan Mukherjee, and Ra{\'u}l
  Rabad{\'a}n.
\newblock Predicting clinical outcomes in glioblastoma: An application of
  topological and functional data analysis.
\newblock {\em Journal of the American Statistical Association},
  115(531):1139--1150, 2020.

\bibitem{Cur12}
Justin Curry, Robert Ghrist, and Michael Robinson.
\newblock Euler calculus with applications to signals and sensing.
\newblock In {\em Proceedings of Symposia in Applied Mathematics}, volume~70,
  pages 75--146, 2012.

\bibitem{CMT18}
Justin Curry, Sayan Mukherjee, and Katharine Turner.
\newblock How many directions determine a shape and other sufficiency results
  for two topological transforms.
\newblock {\em Transactions of the American Mathematical Society, Series B},
  9(32):1006--1043, 2022.

\bibitem{ghrist2018persistent}
Robert Ghrist, Rachel Levanger, and Huy Mai.
\newblock Persistent homology and euler integral transforms.
\newblock {\em Journal of Applied and Computational Topology}, 2:55--60, 2018.

\bibitem{GR11}
Robert Ghrist and Michael Robinson.
\newblock {E}uler{\textendash}{B}essel and {E}uler{\textendash}{F}ourier
  transforms.
\newblock {\em Inverse Problems}, 27(12), 2011.

\bibitem{Hacquard2023}
Olympio Hacquard and Vadim Lebovici.
\newblock Euler characteristic tools for topological data analysis.
\newblock arXiv preprint:2303.14040, 2023.

\bibitem{Lebovici22}
Vadim Lebovici.
\newblock Hybrid transforms of constructible functions.
\newblock {\em Foundations of Computational Mathematics}, pages 1--47, 2022.

\bibitem{eucalc}
Hugo Passe.
\newblock \texttt{eucalc}.
\newblock \url{https://github.com/HugoPasse/Eucalc}, 2021.

\bibitem{S91}
Pierre Schapira.
\newblock Operations on constructible functions.
\newblock {\em Journal of Pure and Applied Algebra}, 72(1):83--93, 1991.

\bibitem{S95}
Pierre Schapira.
\newblock Tomography of constructible functions.
\newblock In {\em International Symposium on Applied Algebra, Algebraic
  Algorithms, and Error-Correcting Codes}, pages 427--435. Springer, 1995.

\bibitem{tang2022sinatra}
Wai~Shing Tang, Gabriel~Monteiro da~Silva, Henry Kirveslahti, Erin Skeens, Bibo
  Feng, Timothy Sudijono, Kevin~K. Yang, Sayan Mukherjee, Brenda Rubenstein,
  and Lorin Crawford.
\newblock A topological data analytic approach for discovering biophysical
  signatures in protein dynamics.
\newblock {\em PLOS Computational Biology}, 18(5):1--42, 05 2022.

\bibitem{T14}
Katharine Turner, Sayan Mukherjee, and Doug~M. Boyer.
\newblock Persistent homology transform for modeling shapes and surfaces.
\newblock {\em Information and Inference: A Journal of the IMA}, 3(4):310--344,
  2014.

\bibitem{V88}
Oleg~Yanovich Viro.
\newblock Some integral calculus based on euler characteristic.
\newblock In {\em Topology and geometry—Rohlin seminar}, pages 127--138.
  Springer, 1988.

\bibitem{Z06}
G{\"u}nter~M Ziegler.
\newblock {\em Lectures on Polytopes, Updated seventh printing of the first
  edition}, volume 152.
\newblock Springer, 2006.

\end{thebibliography}

\appendix

\section{Running example}
\label{app:running-example}

In this section, we present an example of run of our algorithms on a small binary cubical complex. The complex is illustrated in \Cref{fig:example_complex}. Cells with non-zero values include $9$ vertices, $9$ edges, and $2$ squares.

\begin{figure}
    \centering
    \includegraphics[width=0.4\textwidth]{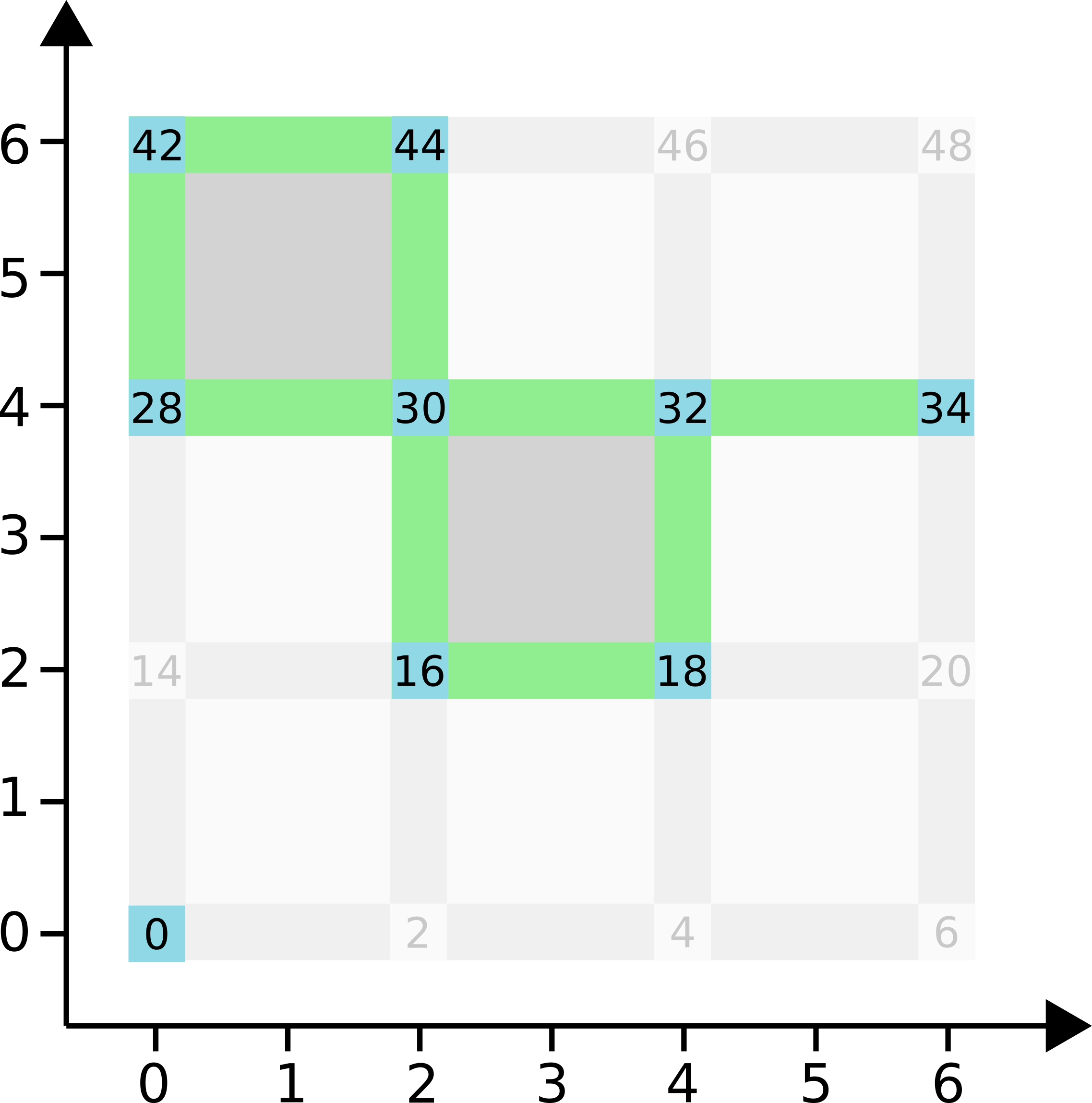}
    \caption{Example of cubical complex. Colored cells have value 1 and light gray cells have value 0. The blue cells are vertices, the green ones are edges and the gray ones are faces. For greater clarity we display the keys of the vertices.\vspace{-1.5em}}
    \label{fig:example_complex}
\end{figure}

\subsection{Preprocessing}
In this section, we unfold the computation of the critical points. We only detail it for the ordinary ones, the computation for the classical ones being similar. We detail the operations for~$\eps = (1,-1)$ and~$x = 32$. First we compute~$e_{-}$, the euler characteristic of the lower star of~$x$ in direction~$\eps$. The coordinates of~$x$ are~$c = (32 \div 7,\, 32 \mod 7) = (4,4)$, so we explore all cells with coordinates~$(4-\delta_0 \eps_0,4-\delta_1 \eps_1)$ where~$(\delta_0 , \delta_1) \in \{0,1\}^2$. These coordinates are~$(4,4),(3,4),(4,5)$ and $(3,5)$. They correspond to the cells~$32,31,39$ and $38$ respectively. Therefore, one has~$e_- = -1 + 1 + 0 + 0 = 0$. In the same way we can compute~$e_+$, the upper star at~$x$ in direction~$\eps$. We visit cells~$32,33,25$ and $26$ and ouput~$e_+ = 1 + -1 + -1 + 0 = -1$. Thus~$e_+ \neq e_-$ and the vertex~$32$ is an ordinary critical point of~$\eps=(1,-1)$ with critical value~$-1$. We list critical points and values for all~$\eps \in \{-1,1\}^2$ in~\Cref{tab:example_critical_points}.

\begin{table}[]
    \centering
    \begin{tabular}{|c||c|c|} \hline 
            \multirow{2}{*}{$\eps$} & $\ordcritpts{\eps}$ &  $\lwcritpts{\eps}$ \\
                                    & $\ordcritvar{\eps}$ &  $\lwcritvar{\eps}$ \\[0.2em] \hline\hline
            \multirow{2}{*}{$(-1,-1)$} & [28, 16, 44, 34] & [0, 30, 44, 34] \\
                                    & [1, 1, -1, -1] & [1, -1, 1, 1] \\[0.2em] \hline
            \multirow{2}{*}{$(-1,1)$} & [42, 18, 32, 34] &  [0, 18, 32, 34] \\
                                    & [1, -1, 1, -1] & [1, 1, -1, 1] \\[0.2em] \hline
            \multirow{2}{*}{$(1,-1)$} & [42, 18, 32, 34] &  [0,42] \\
                                    & [-1, 1, -1, 1] & [1,1] \\[0.2em] \hline
        \multirow{2}{*}{$(1,1)$} & [28, 16, 44, 34] &  [0, 28, 16, 30] \\
                                    & [-1, -1, 1, 1] & [1, 1, 1, -1] \\[0.2em] \hline
    \end{tabular}
    \caption{Critical points and values of the complex in \Cref{fig:example_complex}.\vspace{-1.5em}}
    \label{tab:example_critical_points}
\end{table}

\subsection{Computation}
Now, we unfold the computation step for Radon transform, ECT, and hybrid transform. Throughout the section, we set~$\xi = (2,2)$ and denote its sign vector simply by~$\eps = (1,1)$. We first describe the computation of the ECT. We start by computing the dot product of~$\xi$ with all the critical points in~$\codelwcritpts_\eps = \{0,28,16,30 \}$. These critical points have coordinates~$[-0.5,-0.5]$, $[-0.5,0.167]$, $[-0.167,-0.167]$ and~$[-0.167,0.167]$ in~$\R^2$. Thus the dot products are~$[-2,0.67,-0.67, 0]$. Then, we sort the critical points by increasing dot product, obtaining the list~\texttt{sorted}$(\codelwcritpts_\eps) = [0,28,16,30]$. Finally, we sum the critical values~$\codelwcritval_{\eps}$ in the order given by the sorting of critical points. We begin with~$T=[-\infty]$ and~$E = [0]$. After visiting the first two vertices,~$0$ and~$28$, we obtain~$T = [-\infty,-2,-0.67]$ and~$E = [0,1,2]$. The next sorted vertex is~$16$, which has the same dot product as~$28$. In that case, we add the variation associated to~$16$ to the last element of~$E$. We obtain~$T = [-\infty,-2,-0.67]$ and~$E = [0,1,3]$. Finally, after adding vertex~$30$ we obtain~$T = [-\infty,-2,-0.67,0]$ and~$E = [0,1,3,2]$. The computation of the Radon transform follows the same process at first but with $\codeordcritpts_\eps$ instead of $\codelwcritpts_\eps$. In that case, $T = [-\infty, -0.67, 0.67, 1.33]$ and $E' = [0.0, -2.0, -1.0, 0.0]$. For each $x \in \codelwcritpts_\eps$, we must then find the greatest $i_0$ such that $T[i_0] \leq \dualdot{\xi}{x}$, and add $\dualdot{\xi}{x}$ to $T_c$ and $E'[i_0] + \codelwcritval_{\eps}[x]$ to $E_c$. We thus get $T_c = [-2.0, -0.67, 0.0]$ and $E_c = [1.0, 0.0, -1.0]$. Finally, computing the hybrid transform is done by summing $-\bkernel(\dualdot{\xi}{x}) \cdot \codeordcritval_\eps[x]$ for $x \in \codeordcritpts_\eps$. In our case, with $\bkernel = \exp$, we obtain $\HT(\xi) \simeq -4.715$.
    
\subsection{Evaluation}
We now unfold the evaluation of the Radon transform and ECT. In this example, we want to evaluate~$\ECT(\xi,t)$ for some~$t\in\R$. For $\xi=(2,2)$, we computed in the above section that~$T = [-\infty,-2,-0.67,0]$ and~$E = [0,1,3,2]$. By a binary search, we find the index~$i$ of the largest value in~$T$ that is smaller than~$t$. Then, we let~$\ECT(\xi,t)=E[i]$. For instance~$t = -0.3$ gives~$i=2$ and~$t=1$ gives~$i=3$. Thus, we get~$\ECT(\xi,-0.3)=3$ and~$\ECT(\xi,1)=2$. The evaluation of the Radon transform is similar. First, we search for indices~$i$ such that~$T_c[i] = t$. In that specific case, the method outputs~$\Radon(\xi,t)=E_c[i]$. Apart from this, the computations are identical.

\subsection{Vectorization}

The vectorization routine of the ECT (similar to the vectorization of the Radon transform) takes as input two bounds $a$ and $b$, and a number $N$  of points. Then, it returns $\ECT(\xi,a + \frac{i(b-a)}{N-1})$ for all $i \in [0,\dots,N-1]$. For instance, if $a=-2.0$, $b=2.0$, and $N=5$, we first iterate over $i=0,\dots,4$. For $i=0$, we have $t_0 = a = -2$, so $\ECT(\xi,t_0) = E[1] = 1$. For $i=1$, we know that $t_1 = -1 > t_0$, so we can start to search at index 1 in $T$. Then, we find that $\ECT(\xi,t_1) = E[1] = 1$. Finally, the values of $\ECT(\xi,t_0)$ for $t_0 \in [-2,-1,0,1,2]$ are $[1,1,2,2,2]$.

\section{Proofs of mathematical results}
\label{app:proofs}
In this section, we prove \Cref{lem:expression-transforms-crit,lem:hyperplane-arrangement} on which our algorithms are built. We first state and prove some necessary lemmas before giving the proofs of our main results. Throughout the section, let $\cplx$ be a polytopal complex in $\R^n$ and let $\xi:\R^n \to \R$ be a linear form that is generic on $\cplx$.

These results and notions introduced in \Cref{sec:preliminaries,sec:Euler-critical-points} are more naturally phrased and proven using Euler calculus \cite{S91,V88}, that is, the integral calculus of constructible functions with respect to the Euler characteristic. For instance, our definition of Euler characteristic $\Euler(\phi)$ of a weighted cubical complex $\phi:\cplx\to\Z$ is equal to the Euler integral \cite[Def.~2.2]{S91} of the constructible function $\widehat{\phi}:\R^n\to\Z$ defined by:
\begin{equation*}
    \widehat{\phi} = \sum_{P\in\cplx} \phi(P)\1_{\relint{P}},
\end{equation*}
where $\relint{Q}$ denotes the relative interior of a polytope $Q\subseteq \R^n$. Similarly, for any $\xi\in\R^n$, the function $\Radon(\xi,\cdot):\R\to\Z$ is equal to the so-called \emph{pushforward} $\xi_*\widehat{\phi}$ in the language of Euler calculus; see \cite[Eq.~(2.5)]{S91}. For the sake of simplicity, we avoided the general definition of constructible functions and of Euler integration in the core of the paper. Along the section, we introduce notions of Euler calculus only when necessary to keep the text accessible to the broadest audience and refer to \cite{S91} for more details.

For any $P\in\cplx$ we denote by $\1_P:\cplx\to\Z$ the weighted cubical complex with value $1$ on $P$ and $0$ on any over cells (even on proper faces of $P$). We call $\1_P$ the \emph{indicating function} of $P$. Then, any weighted cubical complex $\phi:\cplx\to\Z$ can be written as a sum~$\phi = \sum_{P\in\cplx}\phi(P)\1_P$. 

The following lemma gives an explicit expression of the Radon transfrom of indicating functions. Before stating it, let us make the following remark. The Euler characteristic, Euler characteristic transforms, Radon transforms and hybrid transforms of a function~$\phi:\cplx\to\Z$ and of its restriction to the complex~$\cplx^\phi$ induced by the collection~$\{P\in\cplx : \phi(P)\ne 0\}$ are equal. 
\begin{lemma}\label{lem:pushforward-polytope}
    Let $P\in\cplx$ and assume that $P$ is not a vertex of $\cplx$. Then, one has~$\Radon[\1_P](\xi,\cdot) = (-1)^{\dim(P)-1}\1_{(\min_P f, \max_P f)}$.
\end{lemma}
\begin{proof}
    This fact is an easy consequence of the Euler calculus formulation of our notions detailed in the beginning of this section. We give an elementary proof for completeness. Let us denote $\phi = \1_P$. Since $\xi$ is linear, the set $P\cap \{\xi=t\}$ is a polytope (possibly empty) for any $t\in\R$. Moreover, since $\xi$ is generic on $\cplx$ and $P$ is not a vertex, this polytope is either a vertex of $P$ or it has dimension $\dim(P)-1$.
    
    Let $t\in\R$. We distinguish three cases. First assume that $t\not\in[\min_P\xi,\max_P\xi]$. Since $P$ is connected and $\xi$ is continuous, one has that $[\min_P\xi,\max_P\xi] = \xi(P)$. Therefore, the polytope $P\cap\{\xi=t\}$ is empty, so is $\cplxeq^\phi$ and hence $\Radon(\xi,t)=0$.

    Assume now that $t\in\{\min_P\xi,\max_P\xi\}$. Since $\xi$ is linear and generic on $\cplx$, it must be that $P\cap\{\xi=t\}$ is a vertex $y$ of $P$. In that case $\cplxeq^\phi = \{y\}$ and $\Radon(\xi,t) = \phi(y)=0$ since $\phi = \1_P$ and $P$ is not a vertex.

    Suppose now that $t\in(\min_P\xi,\max_P\xi)$. In that case, the polytope $Q = P \cap \{\xi=t\}$ is neither empty nor a vertex. Therefore, this polytope has dimension $\dim(P)-1$. Since $\cplx^\phi$ consists of $P$ and all its faces, the polytopal complex $\cplxeq^\phi$ contains only one polytope of dimension $\dim(P)-1$ which is $Q$. Any other polytope $Q'$ of $\cplxeq^\phi$ has dimension smaller or equal to $\dim(P)-2$. In such a case, one has $\phieq(Q') \ne \phi(P)$ and thus $\phieq(Q') = 0$. As a consequence, one has $\Radon(\xi,t) = \phieq(Q) (-1)^{\dim(Q)} = (-1)^{\dim(P)-1}$. Hence the result.%
\end{proof}

For any vertex $x$ of $\cplx$, we denote by $\lwstpt{\xi,x}$ the set $\lwst{\xi,x}\setminus\{x\}$. Similarly, we define the set $\upstpt{\xi,x}$. Moreover, for any function $h:\R\to\R$ we denote one-sided limits at any $t\in\R$ by:
\begin{equation*}
    h(t^-) = \limit_{\substack{s\to t\\s<t}} h(s) \hspace{1cm} \mbox{and} \hspace{1cm}
    h(t^+) = \limit_{\substack{s\to t\\s>t}} h(s),
\end{equation*} 
whenever they are defined. Moreover, following the notations of the beginning of this section, Euler calculus ensures that for any $\phi:\cplx\to\Z$ the function $\Radon(\xi,\cdot) = \xi_*\widehat{\phi}$ is a finite sum of indicating functions of real intervals \cite[Thm.~2.3.(i)]{S91}, hence admits all one-sided limits both from above and from below. The following lemma provides additional details on the Radon transform of indicating functions. It will be crucial in our study of critical points (\Cref{lem:critical-values}). 
\begin{lemma}\label{lem:pushforward-relint}
    Let $x\in\vert{\cplx}$ be such that $\dualdot{\xi}{x}=t$ and let $P\in\str{x}$. One has 
    \begin{equation*}
        \Radon[\1_P](\xi,t) = 
        \begin{cases}
            1                   &\hspace{-2.2em}\big(P = \{x\}\big),\\
            0                   &\hspace{-2.2em}\big(P\in\lwstpt{\xi,x} \cup \upstpt{\xi,x}\big),\\
            (-1)^{\dim(P)-1}    &\!\!\!\!\!\big(P\not\in\lwst{\xi,x} \cup \upst{\xi,x}\big).
        \end{cases}
    \end{equation*}
    Moreover,
    \begin{equation*}\label{lem:pushforward-relint}
        \Radon[\1_P](\xi,t^-) = 
        \begin{cases}
            0                   &\big(P\in\upst{\xi,x}\big), \\
            (-1)^{\dim(P)-1}    &\big(P\not\in\upst{\xi,x}\big).
        \end{cases}
    \end{equation*}
    Similarly,
    \begin{equation*}
        \Radon[\1_P](\xi,t^+) = 
        \begin{cases}
            0                   &\big(P\in\lwst{\xi,x}\big),\\
            (-1)^{\dim(P)-1}    &\big(P\not\in\lwst{\xi,x}\big).
        \end{cases}
    \end{equation*}
\end{lemma}

\begin{proof}
    Throughout the proof, denote $\phi = \1_P$. Let us first prove the first equality. The case $P = \{x\}$ is clear. If $P\in\lwstpt{x}\cup\upstpt{x}$, then $\cplxeq^\phi = \{x\}$ and $\Radon(\xi,t) = \phi(x) = 0$. Now, suppose that $P\not\in\lwst{\xi,x}\cup\upst{\xi,x}$. Since $t \in \xi(P)$, the polytope $Q = P \cap \{\xi=t\}$ is non-empty. Moreover, the polytope $Q$ is not a vertex, for otherwise the generic linear form $\xi$ would be constant on the edge between the distinct vertices $Q$ and $x$. Therefore, the polytope $Q$ has dimension $\dim(P)-1$. As in the proof \Cref{lem:pushforward-polytope}, one can thus show that $\Radon(\xi,t) = \phieq(Q) (-1)^{\dim(Q)} = (-1)^{\dim(P)-1}$.

    Now, let us prove the second equality. If $P\in\upst{\xi,x}$, then $P\subseteq \{\xi\geq t\}$ and $\{\xi = s\}\cap P = \varnothing$ for all $s<t$. Hence $\Radon(\xi,t^-) = 0$. If $P\in\str{x}\setminus\upst{\xi,x}$, then $\min_P \xi < t$. Since both $\min_P \xi$ and $t$ lie in the image of the connected set $P$ by the continuous map $\xi$, then $[\min_P \xi, t]\subseteq \xi(P)$. Therefore, for any $s\in[\min_P \xi, t]$, the polytope $Q_s = \{\xi=s\}\cap P$ is non-empty. The same argument as in the previous paragraph shows that $Q_s$ is not a vertex and hence has dimension $\dim(P)-1$. Once again, we can conclude that $\Radon(\xi,t) = \phieq(Q_s) (-1)^{\dim(Q_s)} = (-1)^{\dim(P)-1}$.

    The proof of the last equality follows from the last one applied to the linear form $-\xi$.
\end{proof}

\begin{lemma}\label{lem:critical-values}
    Let $\phi:\cplx\to\Z$. For any $t\in\R$,
    \begin{align*}
        \Radon(\xi,t) - \Radon(\xi,t^-) &= \sum_{\substack{x\in\lwcritpts{\xi}\\\dualdot{\xi}{x} = t}} \lwcritvar{\xi,x},\\[0.5em]
        \Radon(\xi,t^+) - \Radon(\xi,t^-) &= \sum_{\substack{x\in\ordcritpts{\xi} \\ \dualdot{\xi}{x}=t}} \ordcritvar{\xi,x}.
    \end{align*} 
\end{lemma}

\begin{proof}
    Let us prove the first equality. Let $t\in\R$. Since $\phi = \sum_{P\in\cplx}\phi(P) \1_{P}$ one has: 
    \begin{equation}\label{eq:sum-classical-1}
            \Radon(\xi,t) - \Radon(\xi,t^-) = \sum_{P\in\cplx}\phi(P) \big(\Radon[\1_P](\xi,t) - \Radon[\1_P](\xi,t^-)\big).
    \end{equation}
    Fix $P\in\cplx$. Suppose that no vertex of $P$ belongs to the level set $\{\xi=t\}$. Since $\xi$ is linear, it implies that $t\ne \min_P f$ and $t\ne \max_P f$. Thus $t\in(\min_P f, \max_P f)$ or $t\in\R\setminus[\min_P f, \max_P f]$ and \Cref{lem:pushforward-polytope} implies that:
    \begin{align*}
        \Radon[\1_P](\xi,t)- \Radon[\1_P](\xi,t^-) &= (-1)^{\dim(P)-1}\bigg(\1_{(\min_P \xi, \max_P \xi)}(t) - \1_{(\min_P \xi, \max_P \xi)}(t^-)\bigg)\\
        &= 0.
    \end{align*}
    Now, suppose there exists a vertex $x$ of $P$ such that $\dualdot{\xi}{x} = t$. \Cref{lem:pushforward-relint} implies that if $P\not\in\lwst{\xi,x}$, then $\Radon[\1_P](\xi,t)- \Radon[\1_P](\xi,t^-) = 0$. Indeed, either $P\in\upstpt{\xi,x}$ and then $\Radon[\1_P](\xi,t) = \Radon[\1_P](\xi,t^-) = 0$, or $P\not\in\lwst{\xi,x}\cup\upst{\xi,x}$ and then $\Radon[\1_P](\xi,t) = \Radon[\1_P](\xi,t^-) = (-1)^{\dim(P)-1}$. Therefore, \eqref{eq:sum-classical-1} can be written:
    \begin{equation}\label{eq:sum-classical-2}
        \Radon(\xi,t) - \Radon(\xi,t^-) = \sum\phi(P) \left(\Radon[\1_P](\xi,t)- \Radon[\1_P](\xi,t^-)\right)
    \end{equation}
    where the sum in the right-hand side is over all polytopes $P\in\cplx$ such that there exists $x\in\vert{P}$ with $\{\xi=t\}$ and $P\in\lwst{\xi,x}$.
    Moreover, for any $P\in\cplx$ there is at most one vertex $x$ of $P$ such that $P\in\lwst{\xi,x}$. Indeed, if there were $x\ne x'$ such that $P\in\lwst{\xi,x}\cap \lwst{\xi,x'}$, then $\dualdot{\xi}{x} = \max_P \xi = \dualdot{\xi}{x'}$. Denote $m = \dualdot{\xi}{x} = \dualdot{\xi}{x'}$. Then $\{\xi = m\}\cap P$ is a face of $P$ that is not a vertex and on which the generic linear form $\xi$ is constant, a contradiction. 

    Therefore, \eqref{eq:sum-classical-2} can be written:
    \begin{equation}\label{eq:sum-classical-3}
        \Radon(\xi,t) - \Radon(\xi,t^-) = \sum_{\substack{x\in\vert{\cplx} \\ \dualdot{\xi}{x} = t}} \sum_{P\in\lwst{\xi,x}} \phi(P) \bigg(\Radon(\xi,t) - \Radon(\xi,t^-)\bigg).
    \end{equation}
    Let $x\in\vert{\cplx}$ such that $\dualdot{\xi}{x} = t$ and $P\in\lwst{\xi,x}$. \Cref{lem:pushforward-relint} implies that $\Radon[\1_P](\xi,t) = 1$ if $P = \{x\}$ and $0$ if $P\in\lwstpt{\xi,x}$. The same lemma implies that $\Radon[\1_P](\xi,t^-) = 0$ if $P = \{x\}$ and $(-1)^{\dim(P)-1}$ if $P\in\lwstpt{\xi,x}$. Therefore, one has:
    \begin{align*}
        \sum_{P\in\lwst{\xi,x}}  \phi(P) \left(\Radon[\1_P](\xi,t)- \Radon[\1_P](\xi,t^-)\right) &= \phi(x) - \sum_{P\in\lwstpt{\xi,x}} (-1)^{\dim(P)-1}\phi(P) \\
        &= \sum_{P\in\lwst{\xi,x}} (-1)^{\dim(P)}\phi(P) \\
        &= \lwcritvar{\xi,x}.
    \end{align*}
    Therefore, \eqref{eq:sum-classical-3} can be written:
    \begin{equation}\label{eq:sum-classical-4}
        \Radon(\xi,t) - \Radon(\xi,t^-) = \sum_{\substack{x\in\vert{\cplx} \\ \dualdot{\xi}{x} = t}} \lwcritvar{\xi,x}.
    \end{equation}
    The set $\lwcritpts{\xi}$ being the set of vertices $x$ with non-zero $\lwcritvar{\xi,x}$, the first equality of the lemma follows.

    Now, let us prove the second equality. Replacing $\xi$ by $-\xi$, one gets from \eqref{eq:sum-classical-4} that:
    \begin{equation}\label{eq:sum-classical-reverse}
        \Radon(\xi,t^+) - \Radon(\xi,t) = \sum_{\substack{x\in\vert{\cplx} \\ \dualdot{\xi}{x}=t}} \Euler[\phi_{|\upst{\xi,x}}].
    \end{equation}
    Hence the result, using \eqref{eq:sum-classical-4}, \eqref{eq:sum-classical-reverse}, and the fact that $\Radon(\xi,t^+) - \Radon(\xi,t^-) = \Radon(\xi,t^+) - \Radon(\xi,t) - (\Radon(\xi,t) - \Radon(\xi,t^-))$.
\end{proof}

\subsection{Proof of \Cref{lem:expression-transforms-crit}}
Let us prove the equality for the Radon transform.
Euler calculus \cite[Thm.~2.3.(i)]{S91} ensures that $\xi_*\widehat{\phi}$ is a finite sum of indicating functions of real intervals. Therefore, it is locally constant on the complement of a finite number of points in $\R$. It is then an easy exercise to check that \Cref{lem:critical-values} yields the result:
\begin{equation}\label{eq:Radon-crit}
    \Radon(\xi,\cdot) = \sum_{x\in\lwcritpts{\xi}} \lwcritvar{\xi,x}\cdot \1_{\{\dualdot{\xi}{x}\}} + \sum_{x\in\ordcritpts{\xi}} \ordcritvar{\xi,x}\cdot \1_{(\dualdot{\xi}{x},+\infty)}.
\end{equation}
By definition of the hybrid transform, integrating~\eqref{eq:Radon-crit} against a locally integrable kernel $\kappa : \R\to\C$ yields the result for $\HT(\xi)$.

The formula for the ECT can be easily shown using the classical \emph{convolution} operation of Euler calculus \cite[Sec.~4]{S91}. This bilinear operator takes two functions~$\theta$ and~$\theta'$ which are finite sums of indicating functions of intervals of~$\R$ and returns a function~$\theta\conv\theta'$ with the same property. One can show that~$\ECT(\xi,t) = \Radon(\xi,\cdot)\conv\1_{[0,+\infty)}(t)$; see for instance \cite[Ex.~7.4, Prop.~7.5]{Lebovici22}. The expression for the ECT follows then from~\eqref{eq:Radon-crit} and the fact that for~$a\in\R$ one has~$\1_{(a,+\infty) }\conv\1_{[0,+\infty)} = 0$ and~$\1_{{a}}\conv\1_{[0,+\infty)} = \1_{[a,+\infty)}$; see \cite[Exs.~4.1, 4.3]{S91}. This concludes the proof of \Cref{lem:expression-transforms-crit}.\hfill $\square$

\subsection{Proof of \Cref{lem:hyperplane-arrangement}}
Let $\ccplx$ be a cubical complex. In that case, any edge $\{x,y\} \in\ccplx$ is such that $x-y$ is colinear to a canonical basis vector of $\R^n$. Let $\xi$ and $\xi'$ be in $\R^n\setminus\{0\}$ with $\sgnvec{\xi} = \sgnvec{\xi'}$. Since $\xi$ and $\xi'$ have same sign vector, we have $\sgn\dualdot{\xi}{x-y} = \sgn\dualdot{\xi'}{x-y}$ for all edge $\{x, y\}\in \ccplx$. It is an easy exercise to check that thus $\xi$ and $\xi'$ induce the same order on $\vert{\ccplx}$. As a consequence, the lower and upper stars of vertices of $\ccplx$ with respect to $\xi$ and $\xi'$ are the same. Moreover, the sets and numbers $\lwcritvar{\xi, x}$, $\ordcritvar{\xi, x}$, $\lwcritpts{\xi}$ and $\ordcritpts{\xi}$ depend only on $\lwst{\xi,x}$ and $\upst{\xi,x}$ (and not on $\xi$), hence the result. \hfill $\square$
\end{document}